\documentclass[12pt]{article}
\usepackage{pseudocode}
\usepackage[margin=1in]{geometry}
\usepackage[pagewise]{lineno}
\usepackage{setspace}
\usepackage{amsmath,bm}
\usepackage{url}
\usepackage{psfrag, graphicx, euscript,amsfonts,amsmath,latexsym,amssymb,latexsym} %,
\usepackage{amsthm}
\usepackage[active]{srcltx}
\usepackage{xcolor}
\usepackage{epstopdf}
\usepackage{psfrag, graphicx, euscript,amsfonts,latexsym,amssymb,latexsym}
\usepackage{epstopdf}
\usepackage{float}
\def\bSig\mathbf{\Sigma}

\title{A Guided FP-growth algorithm for multitude-targeted mining of big data}

\usepackage{authblk}

\author[1]{Lior Shabtay \thanks{lior.shabtay@yahoo.com}}
\author[1,2]{Rami Yaari \thanks{ramiyaari@gmail.com}}
\author[1]{Itai Dattner\thanks{idattner@stat.haifa.ac.il}}

\affil[1]{Department of Statistics, University of Haifa, 199 Abba Khoushy Ave, Mount Carmel, Haifa 3498838, Israel}
\affil[2]{Bio-statistical and Bio-mathematical Unit, The Gertner Institute for Epidemiology and Health Policy Research, Chaim Sheba Medical Center, Tel Hashomer, 52621 Israel}

\begin{document}

\maketitle
%\linenumbers

\begin{abstract}
%Insert your abstract here. Include keywords, PACS and mathematical
%subject classification numbers as needed.
In this paper we present the GFP-growth (Guided FP-growth) algorithm, a novel method for multitude-targeted mining: finding the count of a given large list of itemsets in large data. The GFP-growth algorithm is designed to focus on the specific multitude itemsets of interest and optimizes the time and memory costs. We prove that the GFP-growth algorithm yields the exact frequency-counts for the required itemsets. We show that for a number of different problems, a solution can be devised which takes advantage of the efficient implementation of multitude-targeted mining for boosting the performance. In particular, we study in detail the problem of generating the minority-class rules from imbalanced data, a scenario that appears in many real-life domains such as medical applications, failure prediction, network and cyber security, and maintenance. We develop the Minority-Report Algorithm that uses the GFP-growth for boosting performance. We prove some theoretical properties of the Minority-Report Algorithm and demonstrate its performance gain using simulations and real data. 
%\keywords{data mining \and itemset discovery \and multitude-targeted mining \and imbalanced data \and constraints \and minority-class rules \and guided FP-growth}
\end{abstract}


\section{Introduction}\label{s:intro}

This paper is concerned with the problem of targeted mining of itemsets from a given dataset (e.g., \cite{kubat2003itemset}, \cite{lavergne2012min}, \cite{fournier2013meit}, \cite{lewis2016enhancements}, \cite{li2006searching}, \cite{yakout2007mining}, \cite{ghanem2011edp}, \cite{ghanem2014towards}). In particular, we develop a
{\it multitude-targeted mining} approach for counting the number of occurrences of a given large list of itemsets in large data.

The main contribution of this paper is a novel procedure for focused mining of an FP-tree, which we call the {\it Guided FP-growth} procedure or in short GFP-growth. The GFP-growth procedure serves for mining the support of multitude pre-specified itemsets from an FP-tree. The procedure mines only the relevant parts of the FP-tree, in a single and partial FP-growth process, concurrently collecting the support of all specified itemsets. 

Due to its popularity, availability, and many actual-use implementations, we chose to base our work on the original FP-growth algorithm (\cite{han2000mining} and \cite{han2004mining}). Furthermore, many of the works that attempt to improve the FP-growth algorithm (using improved data structures, improved procedures for mining them, parallelization, etc., see e.g.,\cite{wang2002top}, \cite{fournier2013meit}, \cite{wei2008effective}) can also be applied to the GFP-growth procedure presented below, providing additional time and memory costs reduction. In addition, using the FP-growth leads us to consider the targeted mining problem as a constraint which specifies the subset of itemsets which are interesting (see e.g. \cite{aggarwal2016frequent}). This gives rise to further potential advantages. 

We demonstrate that the GFP-growth procedure is a very fast and generic tool, which can be applied to many different purposes. In particular, we study in detail the problem of mining minority-class rules from imbalanced data, a scenario that appears in many real-life domains such as medical applications, failure prediction, network and cyber security, and maintenance. We develop the Minority-Report Algorithm that uses the GFP-growth for boosting performance. We prove some theoretical properties of the Minority-Report Algorithm and demonstrate its performance gain using simulations and real data. 

Further emphasizing the usefulness and applicability of our efforts to improve multitude targeted mining are the works \cite{li2006searching}, \cite{yakout2007mining}, \cite{ghanem2011edp}, \cite{ghanem2014towards}. Indeed, these works take advantage of targeted mining in itemset tree, as a tool used by algorithms for solving problems like frequent itemset mining and association rule mining. Each of these algorithm makes multitude calls to targeted
mining, hence the need for an improved  multitude targeted mining algorithms. 

The paper is organized as follows. In Section~\ref{s:fim} we briefly review related literatures. In Section~\ref{s:gfp} we introduce the Guided FP-growth algorithm and explore its theoretical properties. In Section~\ref{s:unbal} we study in detail the problem of mining minority-class rules from imbalanced data. Section~\ref{s:app} includes a variety of additional examples for which the GFP-growth algorithm can be applied in order to potentially boost performance. A summary of our main contributions, and suggestions for future research are presented in Section~\ref{s:con}. 

\section{Related work}\label{s:fim}

Previous papers studying targeted mining (\cite{kubat2003itemset}, \cite{yakout2007mining}, \cite{fournier2013meit}) use an 'itemset tree' which is similar in concept to the FP-tree (\cite{han2000mining}, \cite{han2004mining}). These algorithms mine the itemset tree given a single target itemset at a time as input. Depending on the items comprising that input target-itemset, each invocation of the procedure potentially traverses a considerable part of the itemset tree. In case a large number of target-itemsets need to be mined, the different respective invocations may overlap, resulting in increased time complexity. In addition, the data for each and every target-itemset is collected from many different locations in the complete itemset tree, again leading to increased complexity, as opposed to collecting data from reduced conditional-trees in FP-growth. 

Some of the works (e.g. \cite{kubat2003itemset}, \cite{lewis2016enhancements}) suggest an itemset mining procedure over itemset trees, in which, given a single itemset as a parameter, all frequent supersets of that itemset are found and counted.

The itemset tree is different from an FP-tree in two aspects. First, it represents the entire dataset, including all items, regardless of their frequency, in this way the itemset tree can be reused for different types of queries, e.g. with different min-support. Second, the item-set tree is arranged according to a pre-defined lexical order of the items. This pre-defined order is actually used by the targeted mining procedures. 
These differences make the FP-tree more compact than an itemset tree, improving the performance due to building the tree from the frequent-items only from the start, and optimizing the order in which the items are used for building the tree. This comes on the expense of an additional pass through the database and less reusability of the tree for potential other required queries. The targeted itemset tree mining procedures also do not involve recursion and creation of conditional sub-trees which are required by the FP-growth procedure, but this comes with the penalty of reduced performance when needing to mine a lot of itemsets, due to the need to traverse the tree separately for each one.

\cite{lavergne2012min} presents an addition of per sub-tree information regarding the minimum and maximum item-id appearing in it (according to the lexical order). This information is used for improving the time complexity of the targeted mining procedure by allowing it to make smarter decisions as to whether it is required to check each sub-tree.
\cite{fournier2013meit} presents a memory-usage related optimization of \cite{kubat2003itemset}, in which each single-prefix path portion of the tree is represented by a single node. Observe that the same optimization can also be applied on an FP-tree.

One can consider the targeted mining problem as a constraint which specifies the subset of itemsets which are interesting (\cite{Pei:2002:CFP:568574.568580}, \cite{Pei2004}, \cite{1316834}, \cite{5066775}, \cite{NGUYEN2015115}, \cite{4603514}, and \cite{aggarwal2016frequent}). Constrained mining is generally used for specifying limitations on the search space in which to look for the itemsets of interest. In most cases, these limitations are manually set, and reflect the preferences of the user. Item constraints are specific type of constraints which specify limitations on the items and item combinations comprising the itemsets of interest. Most works focus on the case in which the itemsets of interest are the supersets of itemsets in the constraint list. Many existing item-constraint related works are based on the Apriori principle of iterative candidate generation and enumeration. Other existing works do not try to address the scalability of constraints number. One FP-growth based work that does address the scalability of item constraints number is \cite{1316834}. Still its scalability is not optimal for our purposes due to checking all the subsets of each itemset-constraint, and due to executing an additional FP-growth in parallel to the one executed over the dataset. \cite{10.1007/978-3-540-24775-3_19} provide a more generic framework, with similar scalability characteristics. 

A different approach is to use a generic algorithm, and adapt it to solve the specific problem at hand. \cite{Pei:2002:CFP:568574.568580} and \cite{Pei2004} provide a generic extention for pattern-growth algorithms, like the FP-growth, supporting convertible monotone and anti-monotone constraints and pushing such constraints deep into the pattern-growth process. One issue with this approach is that the multitude-targeted itemset mining problem is not convertible monotone and not convertible anti-monotone: for example, a convertible anti-monotone procedure, adapted to mine the target itemsets has to include all their prefixes as well. Therefore, such a procedure will explore, count, and output all the target itemsets as well as all their prefixes, which will then need to be removed e.g. at a postprocessing stage. Additional penalty comes from the generic nature of the solution: involving multitude calls to the anti-monotone boolean function (where the complexity of each call depends on the respective itemset length), and missing the problem-specific optimizations and adaptations, such as more focused search and data-reduction.

In the next section we describe the Guided FP-growth procedure, which is specifically adapted and optimized for mining multitude targeted-itemsets. 
\section{The Guided FP-growth procedure}\label{s:gfp}

\subsection{Background}
As mentioned above, the Guided FP-growth procedure developed in this work is based on the original FP-growth algorithm, so we first briefly review some subjects in the context of mining frequent itemsets and association rules. In particular, note that the classical data-mining algorithms such as Apriori \cite{agrawal1994fast} and FP-growth (\cite{han2000mining}, \cite{han2004mining}) are designed for finding \textbf{all} frequent itemsets that comply to a specified minimum support. The problem of multitude-targeted mining requires focusing on a specific set of itemsets, and thus the performance of such algorithms is not optimal in this case, especially when the number of different items is large, the number of items per transaction is large, and the minimum-support is low. 

The Apriori algorithm of \cite{agrawal1994fast} is one of the most popular methods for itemset mining. Many studies adopt an Apriori-like approach, which is based on an anti-monotone heuristic, as presented in \cite{agrawal1994fast}: if an itemset of length $k$ is not frequent (i.e. its support is less than the so called {\it min-support}), the support of any of its length $k+1$ supersets is not frequent as well. The Apriori algorithm iteratively generates candidate itemsets of length $k+1$ from the set of the frequent itemsets of length $k$, and checks their corresponding support in the database. Apriori-like algorithms have two main drawbacks which may degrade their performance: the number of candidate itemsets at each iteration may become huge; the database needs to be scanned at each iteration, where the number of iterations is equal to the length of the longest mined frequent itemset.

The FP-growth (frequent-pattern growth) algorithm, presented in \cite{han2000mining} and \cite{han2004mining} does not require candidate generation and scans the database twice, thus resolving the two drawbacks of the Apriori approach. 

The FP-growth algorithm consists of two main steps: First, the content of the database is compressed into a special tree structure called {\it FP-tree}. Then, the FP-tree is mined by an algorithm called FP-growth to extract all frequent itemsets. 

The FP-tree is constructed in two steps: in the first step, the database is scanned in order to decide which of the items is frequent. Only frequent items are used when building the FP-tree. In the second step the database is scanned again and the actual tree is constructed. The tree nodes are arranged according to the represented-item frequency, such that nodes representing more frequent items have a better chance to be merged.

The FP-growth step of the algorithm mines the FP-tree and extracts the frequent itemsets. The process uses a divide-and-conquer recursive approach: it loops through the frequent items in ascending-support order, and for each such item, it first reports its support and then recursively mines the conditional FP-tree which represents only transactions containing that item.

%As discussed above, 
The FP-growth algorithm may still have high run-time and memory costs. When used for solving the problem of multitude-targeted mining, we have an opportunity to reduce the costs, since the target itemsets are specified, which opens the door to devise a more focused mining effort. Below we present the GFP-growth algorithm, which reduces the time required for the procedure to complete as well as the memory costs, by focusing the efforts to what is required for mining the targets of interest.

For this purpose, we suggest implementing an itemset tree, which we call {\it TIS-tree}, representing the pre-specified itemsets for which we would like to mine the support from the {\it FP-tree}. The {\it TIS-tree} is built according to the order of the pattern-growth process, which in the case of FP-growth is the reverse item-ordering of the one used for building the {\it FP-tree} (i.e. support-ascending order). The {\it TIS-tree} does not need to include itemsets which are known in advance not to appear in the {\it FP-tree}, e.g. those containing items which do not appear in the FP-tree. Each node of the {\it TIS-tree} contains a flag, called 'target', indicating whether this node represents a target itemset.

The following are the highlights of the GFP-growth algorithm, and the way it uses the {\it TIS-tree} for solving the multitude targeted-itemsets mining problem:
\begin{enumerate}
\item
The flow of the GFP-growth procedure follows the {\it  TIS-tree} at each step of its execution. In this way, the GFP-growth procedure actually performs a partial walk over {\it TIS-tree}, focusing the process on the target-itemset tree, which is generally smaller than the {\it FP-tree}.
\item
Since the algorithm follows the {\it  TIS-tree}, it consults the {\it FP-tree} header-table (in $O(1)$) before diving into the creation and exploration of the respective conditional {\it FP-tree} and sub {\it  TIS-tree}. 
\item
In case the next {\it TIS-tree} node to process is a leaf (i.e. has no children), no conditional {\it FP-tree} is created and no recursive call is conducted. This check is fast (O$(1)$) with the right data structure - for example having such a flag in each of the {\it TIS-tree} nodes.
\item
The algorithm uses information about the sub-tree of the {\it TIS-tree} in order to conduct data-reduction when building the conditional {\it FP-tree}. Specifically, the conditional {\it FP-tree} does not need to include items which do not appear in the respective {\it TIS-tree} sub-tree. The optimized conditional {\it FP-tree} creation procedure skips these items when creating the conditional {\it FP-tree}.
\item
We let the GFP-growth procedure provide its results by updating a counter, representing the itemset frequency-count, inside each of the respective nodes of {\it TIS-tree}. This saves the work required for building a separate data strucuture for holding the results.
\item
The following optimization takes advantage of the fact that in many applications the GFP-algorithm is required not to apply a min-support constraint (and should only be applied in these applications): in case the node representing the currently processed prefix is not a target node, there is no need to calculate the support for the itemset represented by that node. For each such prefix, this enhancement eliminates the need for the count-calculation work, which potentially requires going through the linked-list of curerntly processed item in the FP-tree.
\end{enumerate}

Observe that the above optimizations are all based on $O(1)$ checks added at different steps of the algorithm, and therefore do not hurt the scalability and performance even in the improbable extreme worst case. For this reason, these optimizations well fit the multitude-targeted mining problem requirements.

\subsection{GFP-growth - detailed}\label{s:gfp1}

The guided FP-growth procedure serves for mining the support of multitude pre-specified itemsets from an {\it FP-tree}.

Let: 
\begin{itemize}
\item 
I = $\{a_1, a_2,\dots, a_m\}$ be a set of $m$ items;
\item 
$DB = \{T_1, T_2, \dots, T_n\}$ a database that contains $n$ transactions, $T_i = \{a_{i1}, a_{i2},\dots, a_{ik}\}$, $a_{ij} \in I$, $k\leq m$ and $i=1,...,n$;
\item 
An itemset $\alpha$ is a set of items $\{a_1, a_2,\dots, a_k\}$, where $k \leq m$;
\item 
The count $C$ (or occurrence frequency) of an itemset $\alpha$, is the number of transactions containing $\alpha$ in $DB$: $C(\alpha) = |{ T_i : \alpha \subseteq T_i }|$;
\item 
The support $S$ of an itemset $\alpha$, is $C(\alpha) / |DB|$.
\end{itemize}

The inputs for a multitude-targeted mining problem are:
\begin{itemize}
\item
An {\it FP-tree} which is the source for the target mining;
\item 
A tree-based data-structure that contains the itemsets which are the mining targets. We call this data-structure {\it TIS-tree} (Target Item-Set Tree). A node in {\it TIS-tree} which represents an itemset $\alpha$ is denoted by {\it TIS-tree}$(\alpha)$.
\end{itemize}

We say that $\alpha \in${\it TIS-tree} if and only if the itemset $\alpha$ was inserted to {\it TIS-tree.}
For each $\alpha \in${\it TIS-tree}, {\it TIS-tree}($\alpha$).{\it g-count} is a counter. This counter is initiated to zero for each $\alpha \in${\it TIS-tree}. The GFP-growth (guided-FP-growth) procedure which is described below updates this counter.
For each $\alpha \in${\it TIS-tree}, {\it TIS-tree}($\alpha$).{\it target} is a boolean flag saying whether this node represents a target itemset, for which data should be collected. 

As discussed above, the GFP-growth procedure takes advantage of a coordinated arrangement of the trees, and performs coordinated exploration of them. Therefore, the {\it TIS-tree} should be built according to the order of the pattern-growth process, which in the case of FP-growth is the reverse item-ordering of the one used for building the {\it FP-tree} (i.e. support-ascending order). In other words, {\it TIS-tree} should be arranged such that for each pair of nodes $a_i$ and $a_j$ in {\it TIS-tree}, where node $a_j$ is a child of node $a_i$, $C(a_j) \geq C(a_i)$. In this way, by following the {\it TIS-tree} in a top-down manner the GFP-growth procedure ensures that the {\it FP-tree} is explored in a bottom-up manner, as done in FP-growth.

The pseudo-code implementing this coordinated {\it TIS-tree} guided {\it FP-tree} exploration is described below. In the pseudo-code:
\begin{itemize}
\item
$a_i \in${\it TIS-tree} means that $a_i$ is a direct child of the root of {\it TIS-tree} - represented by a node in {\it TIS-tree} denoted by {\it TIS-tree}$(a_i)$
\item
$a_i \in${\it FP-tree} means that $a_i$ appears in the header table of {\it FP-tree}, as defined in \cite{han2004mining}.
\item
'$a_i.${\it count} in {\it FP-tree}' means the count of $a_i$ in the database represented by {\it FP-tree}. The implementation of getting this count from {\it FP-tree} is as described in \cite{han2004mining}, which is to follow the linked list starting at the entry of $a_i$ in the header table of {\it FP-tree}, and summing the counts from the visited nodes
\item
The outcome of the procedure is that each entry in {\it TIS-tree} is updated with the frequency of appearance in FP-tree of the respective itemset 
\end{itemize}

Note that the GFP-algorithm code below does not assume a min-support constraint, and indeed such a constraint is not required by the Minority-Report Algorithm use case and other use-cases as suggested in the sequel. The min-support constraint can be added to the algorithm, just as done in \cite{han2004mining}, \cite{Pei:2002:CFP:568574.568580}, and \cite{Pei2004}, and if added, will affect the created conditional-trees, further reducing their size. 

%\begin{pseudocode}[shadowbox]{GFP-growth}{\text {\it TIS-tree}, {\text {\it FP-tree}}}
%\begin{pseudocode}[framebox]{GFP-growth}{\text {\it TIS-tree}, {\text {\it FP-tree}}}
\begin{pseudocode}[ruled]{GFP-growth}{\text {\it TIS-tree}, {\text {\it FP-tree}}}

\FOR \text{each item } a_i \in \text {\it TIS-tree} \DO
      \IF (a_i \in \text {FP-tree}) \THEN \BEGIN
		\IF (\text {\it TIS-tree}(a_i) \text {\it .target}) \THEN
          \text {\it TIS-tree}(a_i) \text {\it .g-count} = a_i. \text{{\it count} in {\it FP-tree}}; \\
		\IF (\text {\it TIS-tree}(a_i) \text { has children}) \THEN \BEGIN
			\text {construct } a_i\text {'s conditional FP-tree {\it c-Tree}};  \\
			\IF \text {c-Tree} \neq \emptyset \THEN \text {call GFP-growth}(\text {\it TIS-tree}(a_i), \text {\it c-Tree}); \\
		\END
\END

\end{pseudocode}

As explained above, the GFP-growth procedure actually performs a partial walk over {\it TIS-tree}, focusing the process on the target-itemset tree, while consulting the {\it FP-tree} header-table (in $O(1)$) before performing each recursive call. For example, let $a_i$ be an item attached to the root of the {\it TIS-tree}. The GFP-growth procedure processes the node $a_i$ in the loop of its first, 'outmost' invocation, and then it creates a conditional tree for $a_i$ and uses it for invocating a recursive call. Let $a_j$ be an item processed in the loop of this recursive call, which means that $a_j$ is attached to $a_i$ in {\it TIS-tree}. Since the order of {\it TIS-tree} is the reverse of the one in the {\it FP-tree}, any transaction containing itemset $(a_i, a_j)$ will be reflected in the {\it FP-tree} as a node $a_j$ which is downstream of an $a_i$ node. Therefore that node will appear in the conditional tree of $a_i$. This, in turn, ensures that all occurrences of $(a_i, a_j)$ are correctly counted.

This pseudo code takes advantage of the {\it TIS-tree}-focus and provides an optimization that in case {\it TIS-tree}$(a_i)$ is a leaf (i.e. has no children), no conditional FP-tree is created and no recursive call is conducted.
An additional optimization which takes the same idea further is to use information about the sub-tree {\it TIS-tree}$(a_i)$ in order to optimize the created conditional FP-tree {\it c-Tree}. Specifically, {\it c-Tree} does not need to include items which are not in {\it TIS-tree}$(a_i)$. The optimized conditional FP-tree creation procedure skips these items when creating the conditional FP-tree. For example, the call for "GFP-growth({\it TIS-tree}$({m})$, ${(f:3, c:3, b:1)}$)" in the example given below, can actually be replaced by "GFP-growth({\it TIS-tree}$({m}), {(f:3})$)" since $b$ and $c$ do not appear in the sub-tree {\it TIS-tree}$({m})$.
Such an optimization requires an implementation of {\it TIS-tree} that holds the required information at each node. An example implementation is to maintain a bit-map at each node of the {\it TIS-tree}, with a per-item bit telling whether it appears in the sub-tree rooted by that node. In case the number of different items is large, this bit-map can become too large. To cope with this issue one can replace the bit-map with a hash-table, a linked-list, or a set of ranges. 
Alternatively or additionally, known techniques which reduce the overhead of the conditional FP-trees can be used, e.g. as in \cite{wang2002top}. 

The GFP-growth procedure is a general-purpose procedure which can be used for different purposes (e.g. as part of the Minority-Report Algorithm described below). Regardless of the context of the specific use-case, the GFP-growth updates {\it TIS-tree}$(\alpha)${\it .g-count} for each $\alpha \in${\it TIS-tree} to be the frequency-count of $\alpha$ in its input {\it FP-tree} (e.g. $FP_0$ in the Minority-Report Algorithm below), and the actual meaning and use of this gathered information differs according to the specific use-case.

The following theorem shows the correctness of the GFP-growth procedure, in the sense that the count information is correctly collected for all target itemsets.

\newtheorem{theorem}{Theorem}
\begin{theorem}
Guided FP-growth correctness

At the end of execution of the GFP-growth procedure, which is given an FP-tree {\it FP} and an {\it TIS-tree} as input, {\it TIS-tree}$(\alpha)${\it .g-count = }$C(\alpha)$ for each $\alpha \in${\it TIS-tree}, where $C(\alpha)$ is the count of the itemset $\alpha$ in the database represented by {\it FP}.
\end{theorem}
Note: observe that the database represented by {\it FP} does not contain infrequent items which might exist in the original database. 

\begin{proof}

Assume that $C(\alpha) > 0$, 
since {\it TIS-tree} is ordered according to the pattern-growth order, the GFP-growth procedure recursively traverses the path of $\alpha$ in the pattern-growth order. As shown in \cite{han2004mining}, this means that each conditional tree created before reaching the step of adding $a_i \in \alpha$ to the tree-condition, contains a representation of all the occurrences of $\alpha$. This, in turn, means that whenever the code asks if ($a_i \in${\it FP-tree}) for an $a_i \in \alpha$, the answer is yes and the process continues.

Therefore, the complete recursion process continues and the node of $\alpha$ is reached. Now, the code "{\it TIS-tree}$(a_i)${\it .g-count }$= a_i.${\it count} in {\it FP-tree}" is executed. At this point the conditional tree contains a representation of all the transactions which are supersets of $(\alpha - a_i)$, so the operation is actually {\it TIS-tree}$(\alpha)${\it .g-count } $= C((\alpha-a_i)+ a_i) = C(\alpha)$.

In case $C(\alpha) = 0$, the above process will either reach a point in which the conditional tree is empty and therefore {\it TIS-tree}$(\alpha)${\it .g-count} will stay $0$ as initiated, or follow the above process till reaching the step in which it assigns {\it TIS-tree}$(\alpha)${\it .g-count } $= C((\alpha-a_i)+ a_i) = C(\alpha) = 0$.
\end{proof}

In the next section, we demonstrate the power of GFP-growth by showing how it enables devising an efficient solution to the common and practical problem of  minority-class association-rules from imbalanced data. 
\section{Case study: mining minority-class rules from imbalanced data}\label{s:unbal}

In this section we study the problem of mining the minority-class association-rules from imbalanced data. Association rules represent relationships among sets of items in a dataset. An association rule $\alpha \rightarrow \beta$, where $\alpha$ and $\beta$ are disjoint itemsets, represents the likelihood of $\beta$ to occur in a transaction containing $\alpha$. In many implementations, the process of mining association rules is decomposed into two separate steps: 
\begin{enumerate}
\item[(\rm i)] Find the frequent itemsets that comply to a specified minimum support. 
\item[(\rm ii)] Use the frequent itemsets to generate association rules that meet a confidence threshold. 
\end{enumerate}

The concept of using itemset mining for the purpose of finding per-class rules is introduced in \cite{ma1998integrating}. 
Mining of class association-rules can be viewed as a special form of association-rule mining. With classification, it is assumed that each record in the database is assigned with a class. For the purpose of using itemset and association-rule mining in order to derive classification rules, each class is assigned with an item representing it, and each transaction contains exactly one class-item. With such a database, a classification rule-set is a subset of the association rules: those with the specified classes as their consequences.
\cite{ma1998integrating} introduces the notion of {\it ruleitems} for generating class association-rules (CARs) for classification purposes. A {\it ruleitem} is defined as a pair $<${\it condset}, $y>$, representing the classification rule {\it condset }$\rightarrow y$, where {\it condset} is an itemset and $y$ is a single item representing a class.

As explained above, we are interested in mining the class association-rules of specific subset of the classes. Mining such rules has a purpose of its own since it can generate interpretable description of the individual classes. Example use-cases in which such reports provide added value are failure detection and root-cause analysis. Such specific-class reports are most insteresting in imbalanced data scenarios. By imbalanced data we mean that the class distribution is unbalanced, so that the class or classes of interest have a considerably lower probability of occurrence (we call such classes 'minority classes' or 'rare classes'). Scenarios of imbalanced data appear in many real-life domains such as medical applications, failure prediction, network and cyber security, maintenance, etc. 

One way to find the class association-rules of the target class is to use one of the known techniques for finding the {\it ruleitems} of \textbf{all} the classes. Such methods generally use association-rule mining with the min-support set to a sufficiently low value provides itemsets that are correlated with the target class. However, in the case of imbalanced data, when applying class-association rule discovery for mining rules for the rare classes, it is required to set the min-support to a very low value. Unfortunately, class-association rule mining suffers from very low performance when the min-support is low due to the large intermediate data structures created and number of itemsets which pass this weak constraint (most of which are not required when mining rare-class rules). Previous works aiming to resolve this issue mainly focus on Apriori-based algorithms and apply remedies like the use of per-class min-support, new interestingness criteria (e.g. \cite{gu2003association},  \cite{arunasalam2006cccs}, \cite{ndour2012classification}), and mining of the optimal rule set as defined in \cite{li2002mining} (see also \cite{gu2003association}, \cite{ndourclassification}, \cite{ndour2012classification}) in order to reduce the amount of candidates created at each level and the amount of produced rules. 
In \cite{NGUYEN2015107} the authors present an algorithm for mining association rules with class constraints, which can be used for focusing on a specific class.

The above-mentioned algorithms involve an iterative candidate generation and enumeration process, and therefore, their performance is derived from Apriori-like characteristics, meaning that the complete data-set is scanned a number of times and that potentially a huge number of candidates are created during at least some of the iterations. Thus, it makes sense that an approach based on FP-growth will do better in that respect. 

For the purpose of mining the minority-class rules from imbalanced data we develop the {\it Minority-Report Algorithm} (MRA) which is based on the GFP-growth algorithm presented above. We prove some theoretical properties of MRA and demonstrate its superior performance using simulations and real data.  
\subsection{Minority-Report Algorithm (MRA): Mining minority-class rules from imbalanced data using the GFP-growth procedure}
In this subsection, we present a novel FP-growth and GFP-growth based procedure, optimized for mining of rules of a single rare-class from imbalanced data. The procedure uses both the original FP-growth of \cite{han2004mining} and the GFP-growth procedure presented above. The main principles of the new procedure are as follows:
\begin{enumerate}
\item Apply a first pass over the dataset in order to pick only the items which are frequent in the rare class (i.e. the number of rare-class transactions which include each of these items is above the required {\it min-support} threshold), and by this gain a smaller tree-based representation of the database, leading to a large reduction in the usage of time and memory resources.
\item Create two FP-trees, one for the rare class and one for the common class, where the rare-class tree is expected to be considerably smaller.
\item Use the FP-growth procedure of \cite{han2004mining} in order to mine the smaller tree and get the frequent itemsets for the target class, which serve as the antecedent part of all potential rules.
\item Use the Guided FP-growth procedure presented above, in order to mine the frequency of the common class for the itemsets discovered in step 3 above. In this way, the mining of the larger tree is focused and therefore much faster.
\end{enumerate}

We choose to focus on a single rare-class, rather than a number of rare-classes, for the purpose of simplicity. The procedure can be extended to mine a number of rare-classes concurrently by either using per class trees or by maintaining per class counters on each node of a single tree representing all the target classes.

A more formal description of the algorithm is as follows. Let:
\begin{itemize}
\item 
A classification-target item $c$, is an item which is used for classifying transactions. In a basic imbalanced-data scenario we assume that transactions containing this item belong to class $c$, and all others belong to class $0$
\item
An association rule is an implication $\alpha \rightarrow \beta$, where $\alpha$ and $\beta$ are itemsets, and $\alpha \cap \beta = \emptyset$
\item
A classification rule is an association rule $\alpha \rightarrow c$, where $\alpha$ is an itemset, and $c$ is a classification-target item
\item
The support of a rule $\alpha \rightarrow c$ is defined as the support($\alpha \cup c$)
\item
The confidence of a rule $\alpha \rightarrow c$ is defined as support($\alpha \cup c$)/support($\alpha$)
\end{itemize}
In the discussion below, we assume a single target class, denoting the target class as '$1$', i.e. $c=$'$1$'.
\begin{itemize}
\item
$DB_1$ is a subset of $DB$ including only the transactions containing item $1$. $DB_1 = \{T_i \in DB : 1 \in T_i \}$
\item
$DB_0$ is a subset of $DB$ including only the transactions which do not contain the item $1$.
$DB_0 = \{T_i \in DB : 1 \notin T_i \}$
\item
The count $C_1(\alpha)$, is the number of transactions containing $\alpha \cup 1$ in $DB$, $C_1(\alpha) = |\{ T_i : \alpha \cup 1 \subseteq T_i \}|$.
Observe that $C_1(\alpha)$ is also the count of $\alpha$ in $DB_1$, $C_1(\alpha )= |\{ T_i \in DB_1: \alpha \subseteq T_i \}|$
\item
The count $C_0(\alpha)$, is the number of transactions containing $\alpha$ and not $1$ in $DB$,
$C_0(\alpha) = |\{ T_i : \alpha \in T_i \ \& \sim (1 \subseteq T_i )\}|$.
Observe that $C_0(\alpha)$ is also the count of $\alpha$ in $DB_0$, $C_0(\alpha) = |\{ T_i \in DB_0 : \alpha \subseteq T_i \}|$
\end{itemize}

The input for the Minority-Report Algorithm are the transaction database $DB$, the target-class, a minimum support threshold $\xi$, $0 \leq \xi \leq 1$, and a minimum confidence threshold {\it minconf}, $0 \leq$ {\it minconf }$\leq 1$. The output of the algorithm is the set of the target-class rules which confirm to the minimum-support and minimum confidence thresholds. Observe that $\xi$ should be smaller than the relative frequency of the target class in $DB$, i.e. $\xi < |DB_1|/|DB|$, or else no rules will be produced.

We assume an implementation of the FP-growth procedure which inserts each discovered frequent-itemset, along with its frequency-count, into {\it TIS-tree}.

The Minority-Report algorithm uses the {\it TIS-tree} in order to store itemsets which are frequent in $DB_1$. After selecting the items to be used for building the trees as described above, it creates two FP-trees, $FP_0$ which is a representation of $DB_0$, and $FP_1$ which is a representation of $DB_1$. It then applies an FP-growth procedure on $FP_1$ in order to set {\it TIS-tree} to hold the set of frequent itemsets in $DB_1$, and then applies the GFP-growth in order to get their frequency count in $DB_0$.

As already defined above, for each $\alpha \in$ {\it TIS-tree}, {\it TIS-tree}$(\alpha)${\it .g-count}  is a counter, which is initiated to zero and set by the GFP-growth procedure. In the Minority-Report algorithm, the GFP-growth procedure is applied on $DB_0$. We show that after GFP-growth is applied, for each $\alpha \in$ {\it TIS-tree}, {\it TIS-tree}$(\alpha)${\it .g-count} $= C_0(\alpha)$. An additional counter for each $\alpha \in$ {\it TIS-tree} is {\it TIS-tree}$(\alpha)${\it .count}. We assume that the FP-growth procedure, applied by the Minority-Report procedure on $DB_1$, inserts the count it calculated for $\alpha$ to {\it TIS-tree}$(\alpha)${\it .count}. Therefore, after FP-growth is applied, for each itemset $\alpha$ which is frequent in $DB_1$, {\it TIS-tree}$(\alpha)${\it .count} $= C_1(\alpha)$. The two counters are used for calculating the confidence for each reported rule.

\begin{pseudocode}[ruled]{Minority-Report}{DB,\xi, \text {\it minconf}}

C^* = \xi \times |DB| \\
\\
\text{// First pass over the dataset} \\
I' = \{a_k \in I : C_1(a_k) \geq C^*\} \\
DB'_0 = \{T'_i = T_i \cap I' : T_i \in DB_0\} \\
DB'_1 = \{T'_i = T_i \cap I' : T_i \in DB_1\} \\
\\
\text{// Second pass over the dataset} \\
FP_0 = \text{FP-tree}(DB'_0) \\
FP_1 = \text{FP-tree}(DB'_1) \\
\\
\text{// Apply the classical FP-growth on } FP_1 \text{, creating {\it TIS-tree}} \\
\text{\it TIS-tree} = \text{FP-growth}(FP_1, \text{min-count} = C^*) \\
\\
\text{// Apply the GFP-growth procedure on } FP_0 \text{, populating {\it TIS-tree}} \\
\text{GFP-growth}(\text{\it TIS-tree}, FP_0\text{)} \\
\\
\text{// Pruning in order to generate strong rules} \\
\text{// using the well-known technique} \\
\FOR \text{each itemset } \alpha \in \text{\it TIS-tree} \DO \BEGIN
     \text{\it conf} = \text{\it TIS-tree}(\alpha)\text{\it .count }/ \\
     \text{        }(\text{\it TIS-tree}(\alpha)\text{\it .count} + \text{\it TIS-tree}(\alpha)\text{\it .g-count}) \\
     \IF (\text{\it conf}\geq \text{\it minconf}) \THEN \BEGIN
           \text{generate-rule } R=(\alpha \rightarrow 1) \\
           \text{\it support}(R)= \text{\it TIS-tree}(\alpha)\text{\it .count}/|DB| \\
           \text{confidence}(R)= \text{\it conf} \\
\END
\END
\end{pseudocode}

The Minority-Report algorithm can be performance-optimized by coordinating the way in which the data-structures are organized. More specifically, a performance-optimized implementation of Minority-Report should:
\begin{itemize}
\item 
Use identical item-ordering for the two FP-trees. For example, use the item-support over the entire $DB$ for determining the item ordering in both FP-trees (support-descending order). Such an item-ordering ensures that the FP-trees used by the algorithm are as-good-as and comparable to that of an FP-growth applied over the entire $DB$.
\item
Build the {\it TIS-tree}, using support-ascending order, as described in Section~\ref{s:gfp1}. 
\end{itemize}

This coordinated order of the trees optimizes the performance of the construction of the {\it TIS-tree} since it is in-line with the order in which the FP-growth procedure traverses $FP_1$. In more details, let $a_i$ be a frequent item that is processed in the loop of the first 'outmost' invocation of FP-growth. The FP-growth procedure will insert it to the {\it TIS-tree} and attach it to the root node. It will then create a conditional tree for $a_i$, containing only items which come after $a_i$ in the decided order of the pattern-growth process and use it for invocating a recursive call. Let $a_j$ be an item processed in the loop of this recursive call, such that $(a_i, a_j)$ is a frequent itemset. The FP-growth procedure simply adds $a_j$ to the {\it TIS-tree} and attaches it to $a_i$. The process continues this way and the {\it TIS-tree} is created in the pattern-growth process order.

\begin{theorem}\label{th:MR}
At the end of the GFP-growth step of the Minority-Report algorithm execution, {\it TIS-tree} contains all and only the itemsets $\alpha$, such that {\it support}$(\alpha \cup 1) = C_1(\alpha)/|DB| \geq \xi$. For each such itemset $\alpha$, {\it TIS-tree}$(\alpha)${\it .count }$ = C_1(\alpha)$ and {\it TIS-tree}$(\alpha)${\it .g-count }$ = C_0(\alpha)$.
\end{theorem}

The claims in this theorem are divided into the following three lemmas, and therefore its correctness is directly derived from their correctness, which is proved below.

\newtheorem{lemma}{Lemma}
\begin{lemma}
At the end of the GFP-growth step of algorithm execution, $\alpha \in${\it TIS-tree} for all the itemsets $\alpha$ for which {\it support}$(\alpha \cup 1) = C_1(\alpha)/|DB| \geq \xi$, and for each such $\alpha$, {\it TIS-tree}$(\alpha)${\it .count }$ = C_1(\alpha)$.
\end{lemma}

\begin{proof}
Observe that the way {\it TIS-tree} is created by the algorithm is by applying the exact three steps of the classical FP-growth algorithm on $DB_1$ with {\it min-count} of $\xi \times |DB|$. Since $C_1(\alpha)$ is the count of $\alpha$ in $|DB_1|$, and since $C_1(\alpha) = C(\alpha \cup 1) \geq \xi \times |DB|$, the correctness of the classical FP-growth leads to the conclusion that $\alpha$ is in TIS-tree and that {\it TIS-tree}$(\alpha)${\it .count }$ = C_1(\alpha) = C(\alpha \cup 1)$.
\end{proof}

\begin{lemma}
At the end of the GFP-growth step of algorithm execution, each of the elements of {\it TIS-tree} represents an itemset $\alpha$ for which {\it support}$(\alpha \cup 1) = C(\alpha \cup 1)/|DB| \geq \xi$.
\end{lemma}

\begin{proof}
Since {\it TIS-tree} is created by applying the exact three steps of the classical FP-growth algorithm on $DB_1$ with {\it min-count} of $\xi \times |DB|$, $C_1(\alpha) \geq \xi \times |DB|$ holds for each itemset $\alpha$ represented by a node in the resulting {\it TIS-tree}. Since $C_1(\alpha) = C(\alpha \cup 1)$, we get $C(\alpha \cup 1) \geq \xi \times |DB|$, and therefore {\it support}$(\alpha \cup 1) = C(\alpha \cup 1)/|DB| \geq \xi$. 
\end{proof}

\begin{lemma}
At the end of the GFP-growth step of algorithm execution, for each itemset $\alpha \in${\it TIS-tree}, {\it TIS-tree}$(\alpha)${\it .g-count } $= C_0(\alpha)$.
\end{lemma}

\begin{proof}
The algorithm mines the FP-tree created for $DB_0$ ($FP_0$) using the guided FP-growth procedure and the TIS-tree as a guide. For each $\alpha \in${\it TIS-tree}, $FP_0$ correctly represents $C_0(\alpha)$ since all items in $\alpha$ were taken into consideration when building $FP_0$ from $DB_0$. This is true since {\it TIS-tree} was extracted from $FP_1$ and the same items were taken into consideration while building $FP_0$ and $FP_1$. Therefore, according to Theorem 1 above (GFP-growth correctness), each {\it TIS-tree}$(\alpha)$.g-count value reported for a rule $\alpha \rightarrow 1$ is equal to $C_0(\alpha)$. 
\end{proof}

\begin{theorem}
Minority-Report Algorithm correctness

The Minority-Report algorithm generates all and only the rules $\alpha \rightarrow 1$ which confirm to the required min-support and min-confidence constraints, where $1$ is the rare class. The support and confidence reported by the algorithm for each of these rules are their support and confidence in $DB$.
\end{theorem}

\begin{proof}
According to Theorem~\ref{th:MR}, at the end of the GFP-growth step of the Minority-Report algorithm execution, TIS-tree contains all and only the itemsets $\alpha$ for which {\it support}$(\alpha \cup 1) = C(\alpha \cup 1)/|DB| \geq \xi$.  Since {\it support}$(\alpha \rightarrow 1) = support(\alpha \cup 1)$, all and only frequent rules $\alpha \rightarrow 1$ are considered for generation at that point in the algorithm. According to the definition of confidence and to Theorem~\ref{th:MR}, the confidence calculated for each rule $\alpha \rightarrow 1$ at the next step of the algorithm is indeed {\it confidence}$(\alpha \rightarrow 1)$. Therefore, the algorithm reports all and only the rules $\alpha \rightarrow 1$ which confirms to the support and confidence constraints, and reports their correct support and confidence attributes. 
\end{proof}

\subsection{Example of Minority-Report Algorithm using guided-FP-growth}
%Generating the pictures: 1. draw again in power point %(pictures_for_example.ppt). 2. choose all the nodes and lines in the picture and right click to save as jpg. 3. convert to eps using http://www.tlhiv.org/rast2vec/ 
This example illustrates an execution of the Minority-Report algorithm, using the GFP-growth algorithm. Let DB be as given in Table~\ref{table:DBMR}. 

Let us  now follow the execution of the Minority-Report algorithm, with DB as the input dataset, class $1$ as the target class, min-support $\xi = 0.125$ (i.e. min-count $C^* = 1$), and min-confidence = $0.2$.
The first pass over the dataset calculates $I^\prime$, which is the list of items that count above $C^*$ in $DB_1$ $(I^\prime = \{a_k \in  I : C_1(a_k)\geq C^*\})$. With the DB shown in Table~\ref{table:DBMR}, $I^\prime= \{f, c, b, m\}$. Items $a, d, g, h, i, l, n, o, p$ are removed due to not passing the threshold in class $1$.
Observe that this step is a significant optimization of the algorithm in both time and memory terms, as the fact that some items are not used for building the FP-trees leads to much smaller trees. Smaller trees mean less memory and less time for their building and mining. Observe also that all these items pass the min-support threshold in DB, so the well-known solution which executes FP-growth for DB would end up with a much larger FP-tree, including all of them.
	
%\begin{table}[h!]
\begin{table}
%\centering
\caption{A database example}
\label{table:DBMR}
%\begin{tabular}{| l| l| l| }
\begin{tabular}{lll}
% \hline
\hline\noalign{\smallskip}
 TID&Items&Class\\
% \hline
\noalign{\smallskip}\hline\noalign{\smallskip}
100& 	f, a, c, d, g, i, m, p &	0\\
200& 	a, b, c, f, l, m, o 	&0\\
300& 	b, f, h, j, o 		&0\\
400& 	b, c, k, s, p 		&0\\
500& 	a, f, c, e, l, p, m, n &	0\\
600&	f, m			&1\\
700& 	c			&1\\
800&	b			&1\\
% \hline
\noalign{\smallskip}\hline
\end{tabular}
\end{table}

The second pass through the dataset creates the two FP-trees, $FP_0$ and $FP_1$.
$FP_1$ is presented in Figure~\ref{fig:ex1}, while $FP_0$ is presented in Figure~\ref{fig:ex2}. 
Next, FP-growth is applied on $FP_1$, creating {\it TIS-tree}. {\it TIS-tree} is presented in Figure~\ref{fig:ex3}. 

%\begin{figure}[!ht]
\begin{figure}
\centering
\includegraphics[width=0.34\textwidth]{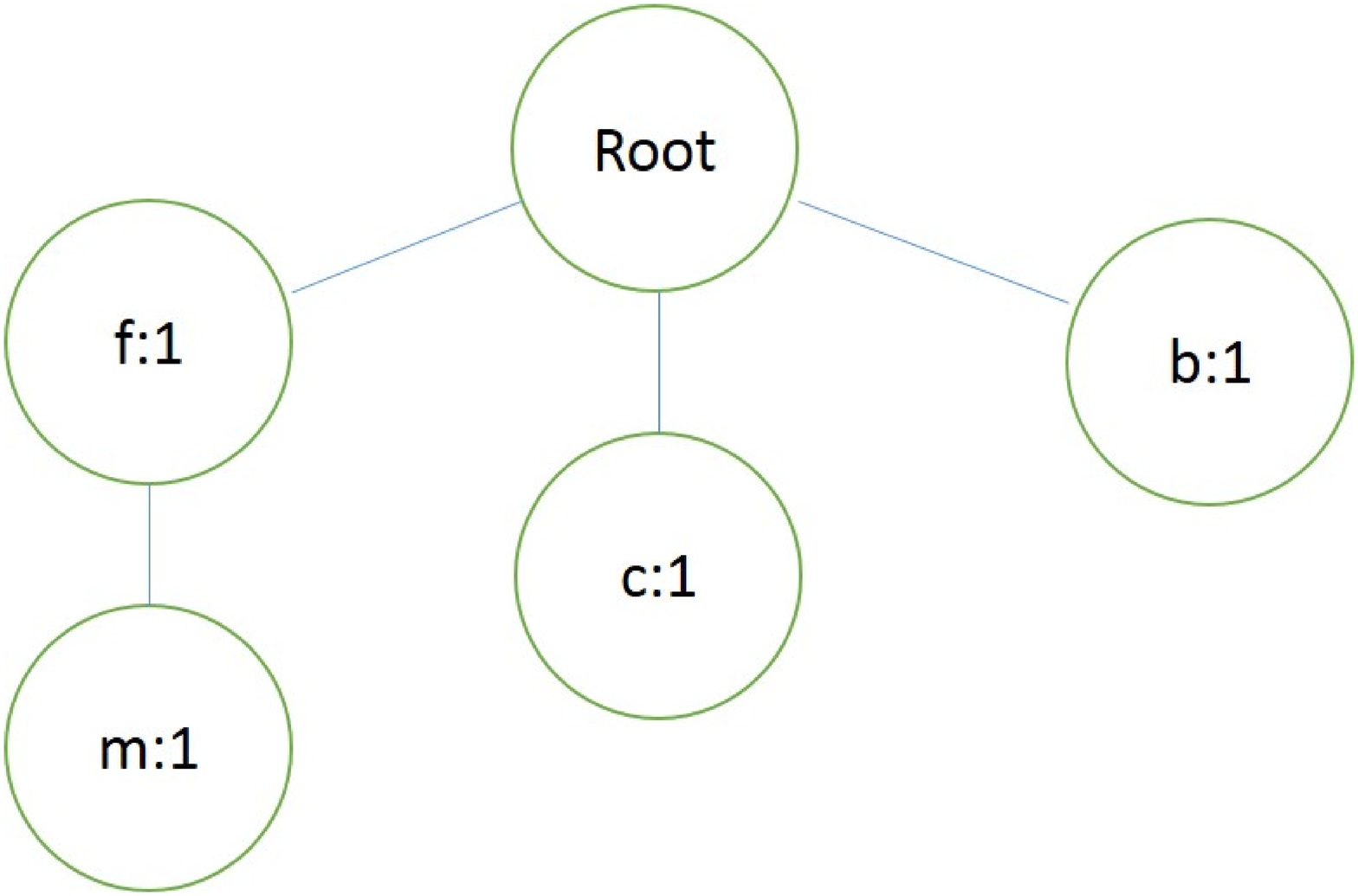}
    \caption{\label{fig:ex1}Tree $FP_1$ of example of Minority-Report Algorithm using guided-FP-growth  
 }
\end{figure}

\begin{figure}
 \centering
\includegraphics[width=0.8\textwidth]{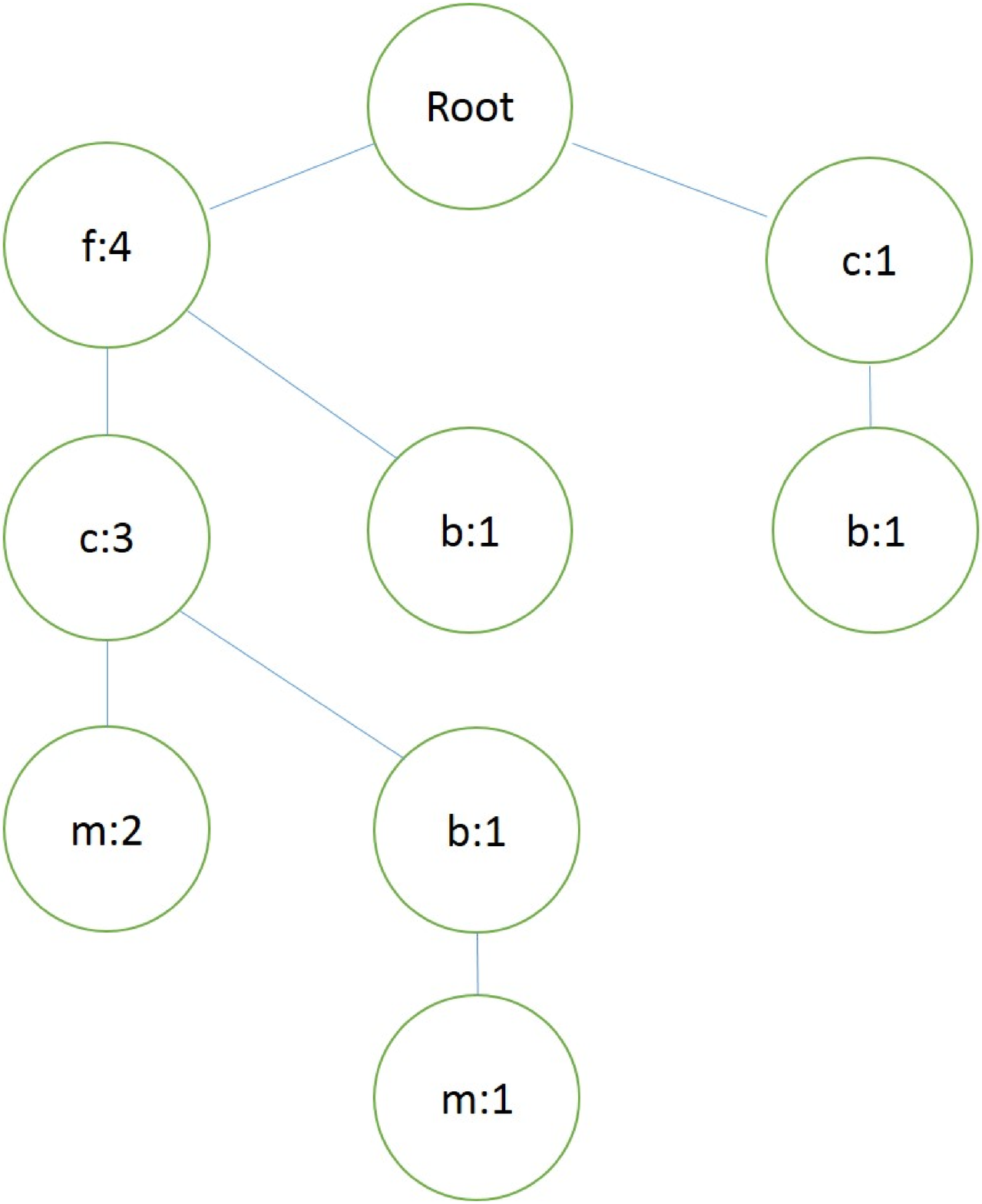}
    \caption{\label{fig:ex2} Tree $FP_0$ of example of Minority-Report Algorithm using guided-FP-growth  
 }
\end{figure}

\begin{figure}
 \centering
\includegraphics[width=0.8\textwidth]{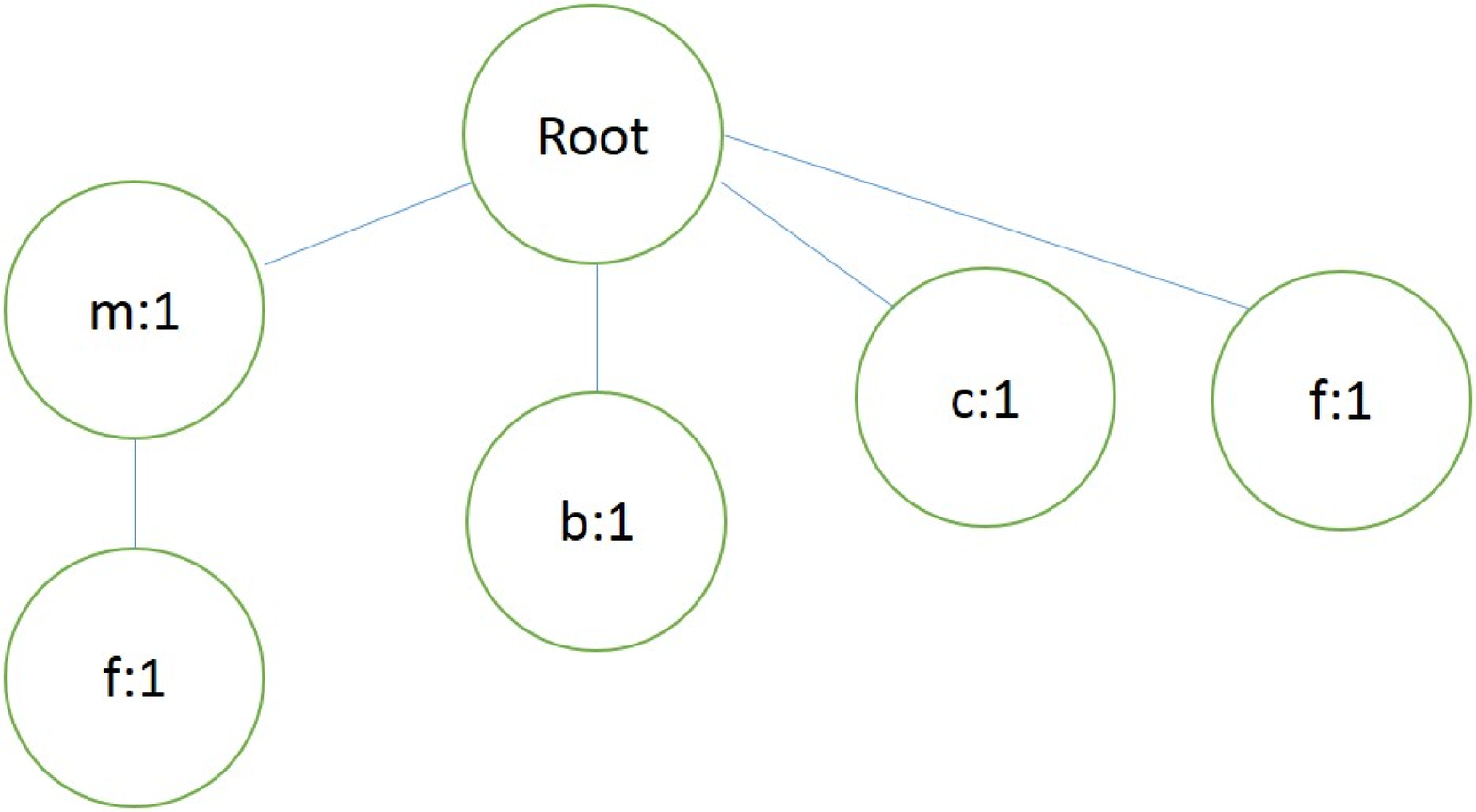}
    \caption{\label{fig:ex3}The {\it TIS-tree} of example of Minority-Report Algorithm using guided-FP-growth  
 }
\end{figure}
Now, the Minority-Report algorithm calls the GFP-growth procedure as follows:
{\it call GFP-growth(TIS-tree, $FP_0$)}. The GFP-growth loops through the items $I^\prime= \{f, c, b, m\}$, as follows:
\begin{enumerate}
\item For $m${: \it TIS-tree}$(\{m\})${\it .g-count }$=3$, {\it c-Tree}$(m)=\{(f:3)\}$ 
call GFP-growth({\it TIS-tree}$(\{m\})$, $\{(f:3)\})$
(observe that $b$ and $c$ are not included in the conditional tree, since they do not appear in {\it TIS-tree}$(\{m\})$) \\
GFP-growth$(${\it TIS-tree}$(\{m\})$, $\{(f:3)\})$ performs a single-iteration loop for $f$ as follows: \\
{\it TIS-tree}$(\{m, f\})${\it .g-count }$=3$ 
(no need for a recursive call since {\it TIS-Tree}$(\{m, f \})$ has no children)
\item For $b${: \it TIS-tree}$(\{b\})${\it .g-count }$=3$
(no need for a recursive call since {\it TIS-Tree}$(\{b\})$ has no children)
\item For $c${: \it TIS-tree}$(\{c\})${\it .g-count }$=4$
(no need for a recursive call since {\it TIS-Tree}$(\{c\})$ has no children)
\item For $f${: \it TIS-tree}$(\{f\})${\it .g-count }$=4$ 
(no need for a recursive call since {\it TIS-Tree}$(\{f\})$ has no children)
\end{enumerate}

Now, {\it TIS-tree} looks as presented in Figure~\ref{fig:ex5}. 

\begin{figure}
 \centering
\includegraphics[width=0.8\textwidth]{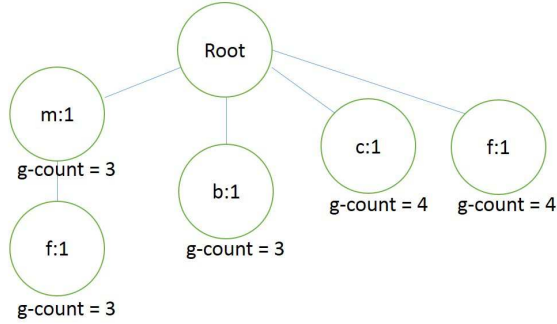}
    \caption{\label{fig:ex5} Second {\it TIS-tree} of example of Minority-Report Algorithm using guided-FP-growth
 }
\end{figure}
At the end of the GFP-growth procedure, {\it g-counter} is updated at all nodes of {\it TIS-tree}.

The last step of the Minority-Report algorithm turns each node in {\it TIS-tree} into a rule, calculates the confidence from the {\it counter} and {\it g-counter} fields of the node, and checks against the min-confidence threshold (which is 0.2 in our example):
\begin{enumerate}
\item Confidence$(\{m\}) = 1/(1+3) = 0.25$
\item	Confidence$(\{b\}) = 1/(1+3) = 0.25$
\item	Confidence$(\{c\}) = 1/(1+4) = 0.2$
\item	Confidence$(\{f\}) = 1/(1+4) = 0.2$
\item	Confidence$(\{m, f\}) = 1/(1+4) = 0.2$
\end{enumerate}

Turns out that all pass the min-confidence threshold and therefore all are turned into rules which are reported along with their respective support and confidence: $\{m\}\rightarrow 1, \{b\} \rightarrow 1, \{c\} \rightarrow 1, \{f\} \rightarrow 1, \{m, f\} \rightarrow 1$.

\subsection{Numerical results}

The computational performance of frequent itemsets mining algorithms is affected  by several parameters: the min-support threshold which is a parameter of the method, the total number of different items in a dataset, the number of transactions, and the statistical characteristics of the data. By 'statistical characteristics' we mean that we view a dataset as a realization of a random process that follows some probabilistic model. For instance, in the simulations below we use a probabilistic model where each item in a transaction is a Bernoulli random variable with success probability $p_X$. Given a probabilistic model, one can determine other data characteristics such as the average or maximum transaction length. This in turn affects the computational performance of frequent itemsets mining algorithms. In general, one can say that the running time of frequent mining algorithms increases as the number of items and number of transactions increases. However, this depends also on the statistical characteristics of the data at hand; see for instance \cite{heaton2016comparing}. 

In this part of our study we are interested in understanding the data scenarios where using the Minority Report algorithm, and hence the GFP-growth algorithm, leads to improvement in computational performance. In this use-case, there is another important parameter, the class distribution. By class distribution we mean the level of imbalance of the response variable, or the target probability which we denote by $p_Y$, where $Y=$'$1$' is the target variable. Specifically, we consider various scenarios where we fix some parameters and calculate the running time of the algorithm. Each such scenario was repeated $20$ times (Monte Carlo experiments) and the average of all running times over the $20$ simulations is reported.  The simulations were run over a Linux Virtual-Machine, m4.16xlarge Amazon Web Services (AWS) instance with 256 GB RAM. We used the $C$ implementation of FP-growth by Christian Borgelt (\cite{borgelt2012frequent}, \cite{borgelt2005implementation}) and altered it to perform the Minority Report algorithm. The performance figures reported in this section are based on a partial GFP-growth implementation.  Specifically it does not reflect the potential saving that would result from the enhancements which save work and reduce data according to the content of the sub TIS-tree, as described above.

Figures~\ref{fig:sim1} (a), (b) and (c) present the running time (in seconds) of the FP-growth (a), GFP-growth (Minority Report; b) and the ratio between their running time (c) as a function of the number of the target-class {\it ruleitems} which appeared in the data. Here $p_X=0.125$, $p_Y=0.01$ and min-support is set to $5\times 10^{-5}$. Each plot includes three lines corresponding to different transaction numbers: $25,000$, $50,000$ and  $100,000$. Each of such lines also delivers information regarding the number of possible different items that ranges from $60$ to $100$. One can see that for a fixed number of transactions, the running time of the algorithms increases as the number of target-class {\it ruleitems} increases, which goes together with increasing number of items and increasing average transaction length. Figure~\ref{fig:sim1} (c) suggests that the GFP-growth is faster than the FP-growth, where the improvement in this case ranges from about $10$ up to $80$ times faster. 

Figures~\ref{fig:sim1} (d), (e) and (f) display a similar scenario, now with $p_Y=0.1$ and min-support of $5\times 10^{-4}$. There is an improvement of GFP-growth but not as substantial as for the case where $p_Y=0.01$ which resembles very unbalanced data.

\begin{figure}
 \centering
\includegraphics[width=0.85\textwidth]{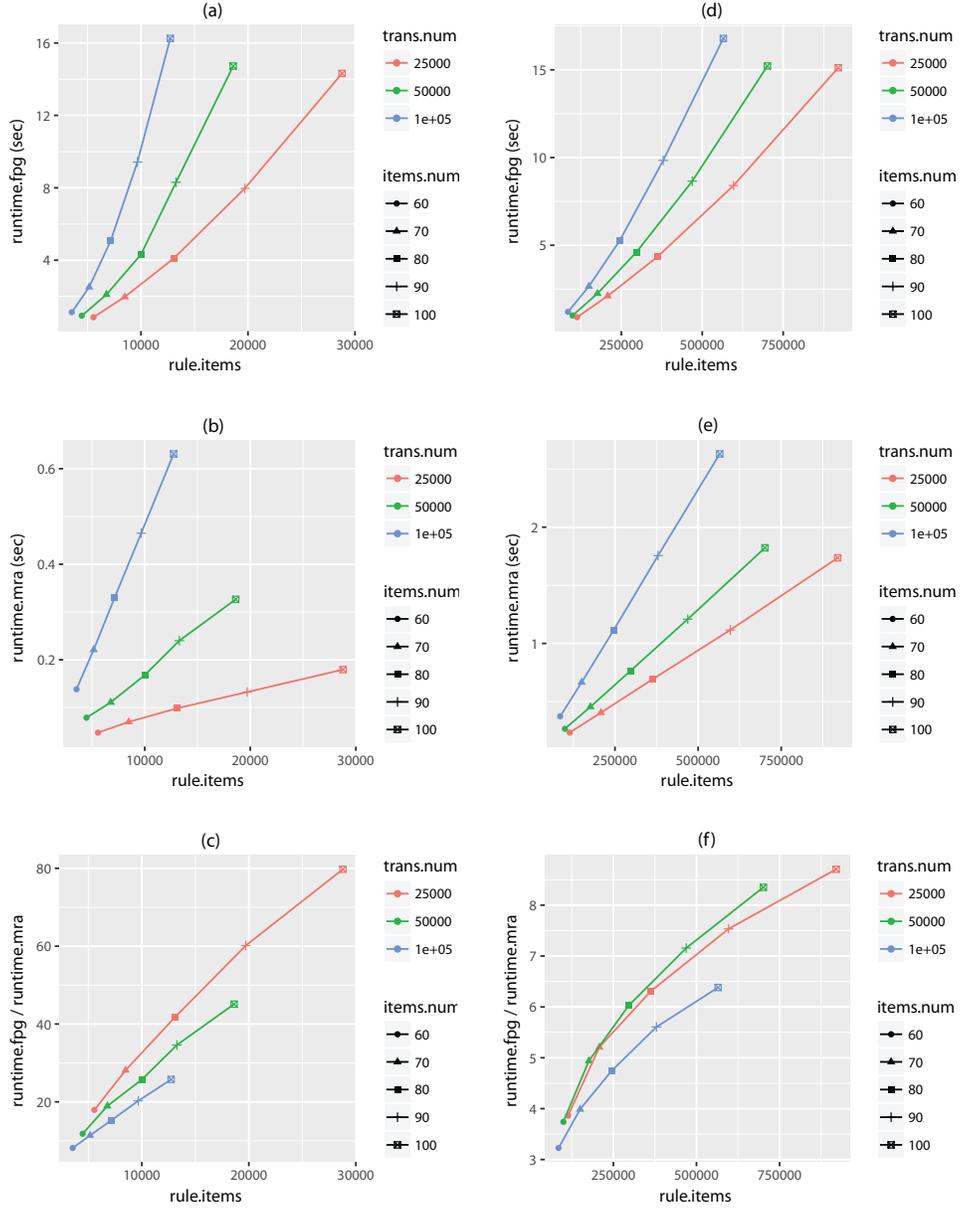}
    \caption{\label{fig:sim1} Simulation results with $p_X=0.125$, $p_Y=0.01$ and min-support of $5\times 10^{-5}$. \textbf{(a)}: running time of FP-growth as a function of number of target-class {\it ruleitems}. \textbf{(b)}: running time of GFP-growth as a function of  number of target-class {\it ruleitems}.  \textbf{(c)}: ratio of running time of FP-growth to GFP-growth as a function of  number of target-class {\it ruleitems}.
Simulation results with  $p_X=0.125$, $p_Y=0.1$ and min-support of $5\times 10^{-4}$. \textbf{(d)}: running time of FP-growth as a function of number of target-class {\it ruleitems}. \textbf{(e)}: running time of GFP-growth as a function of  number of target-class {\it ruleitems}.  \textbf{(f)}: ratio of running time of FP-growth to GFP-growth as a function of  number of target-class {\it ruleitems}  
 }
\end{figure}

In Figures~\ref{fig:real} we present results of applying the algorithms to real data. We consider the 'Census income' dataset from \url{http://archive.ics.uci.edu/ml/datasets/Adult}. We removed rows with missing information as well as the columns: 'capital.loss', 'capital.gain' and 'education.num'.
The 'fnlwgt' column was discretized into a categorical variable with four categories. The age and hours.per.week columns were discretized into categorical variables of five and six categories respectively: age=('17-25','26-35','36-45','46-65','66+'), hours.per.week=('1-10','11-20','21-30','31-40','41-50','51+').
The number of columns in the database after our manipulations are 12, with each column taking several possible categories, summing up to 115 possible different items. For the target we use the variable 'salary',  which can be either $\leq 50K$ or $>50K$. The resulting database had approximately $30,000$ rows, with a  distribution of the target class being $75\%$ $salary\leq 50K$ and $25\%$ $salary>50K$.  In order to create an imbalanced problem, we sampled for each test $22,500$ rows with the number of rows with salary more than $50K$ set to $22,500 \times p_Y$, and the rest of the rows had salary less than or equal to $50K$. The results presented in Figures~\ref{fig:real} summarize the mean of $20$ such test samplings for each scenario. As in the simulation results above, our main observation is that the target probability has the most substantial effect on the running time of the GFP-growth algorithm and hence its improvement over the FP-growth which can be as high as $50$ times faster. 

\begin{figure}
 \centering
\includegraphics[width=0.50\textwidth]{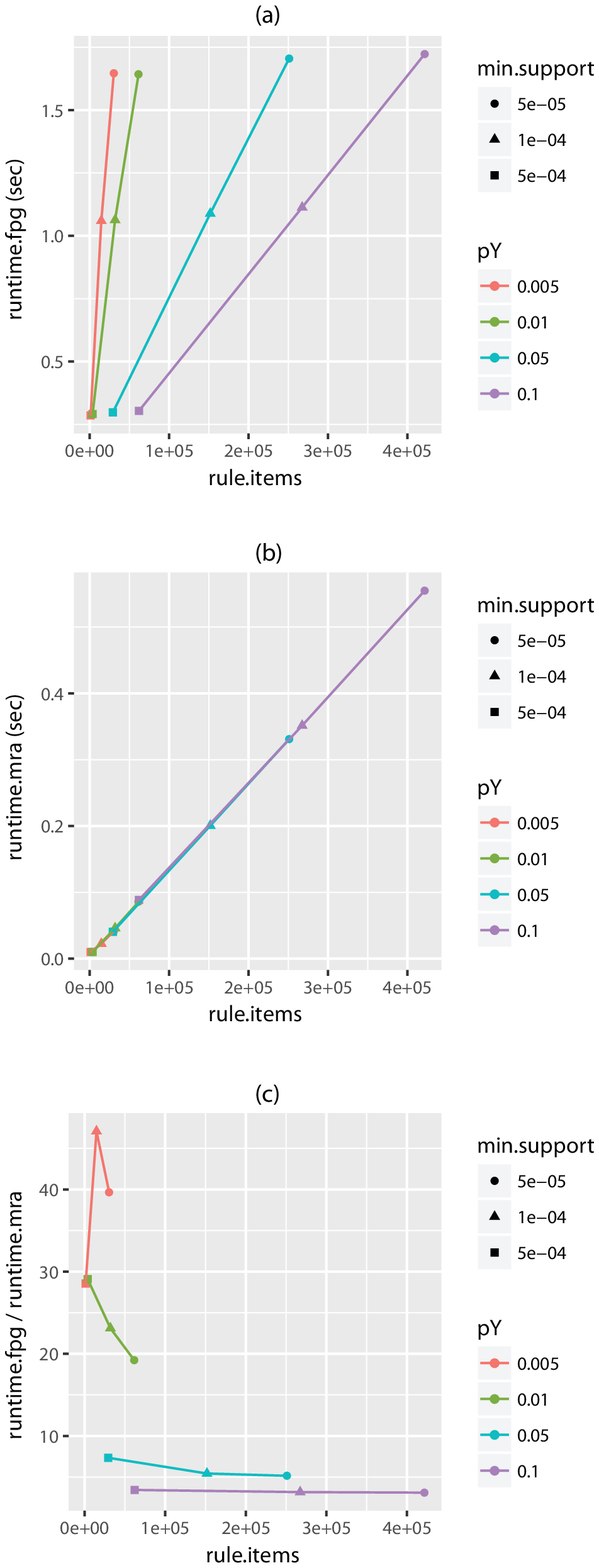}
    \caption{\label{fig:real} Real data results. \textbf{(a)}: running time of FP-growth as a function of number of target-class {\it ruleitems}. \textbf{(b)}: running time of GFP-growth as a function of  number of target-class {\it ruleitems}. \textbf{(c)}: ratio of running time of FP-growth to GFP-growth as a function of  number of target-class {\it ruleitems}  
 }
\end{figure}

Note that the min-support threshold is a tuning parameter of the method and its 'optimal' choice should be defined according to the question of interest. For instance, for a supervised learning problem such as classification, it makes sense that the min-support threshold will control the bias-variance trade-off of a method and as a results will affect its accuracy (generalization ability). Such an observation was made in \cite{coenen2005obtaining}. In our case we are concerned with an exploratory data analysis and are interested in mining \textbf{all} ruleitems so an 'optimal' choice of min-support is not relevant.
%
%ã

%
%
\section{Extensions and future work}\label{s:app}

In the previous section we demonstrated the application of the GFP-growth algorithm to a specific use case: mining {\it ruleitmes} from imbalanced data. However, the GFP-growth can be used for solving a various of mining problems, some examples follow. 

\subsection{Mining frequent itemsets and optimal classification-rules}

\cite{li2006searching} presents an algorithm which follows the principle of iterative candidate generation process (\cite{agrawal1994fast}).  The paper uses the itemset tree of \cite{kubat2003itemset} in the process of counting the frequency of the candidate itemsets. In particular, it introduces an enhancement in the context of their algorithm for frequent-itemset mining, which is to maintain a per candidate list which records all occurrences of the candidate in the itemset tree, and use the list for mining its direct children candidates. It requires collecting the data for each and every target-itemset from many different locations in the complete itemset tree.

\cite{yakout2007mining} presents an algorithm which follows the Apriori principle. The paper uses the itemset tree of \cite{kubat2003itemset}, and \cite{fournier2013meit} in the process of counting the frequency of the candidate itemsets. \cite{yakout2007mining} describes and uses a procedure which uses the itemset tree for counting frequency of the candidate itemsets which share the leading items and differ only in the last item. The procedure is called per discovered frequent itemset, and depending on the items comprising that itemset, each invocation of the procedure potentially traverses a considerable part of the itemset tree. The different invocations may overlap, meaning that the same parts of the itemset tree are mined a number of times during the different invocations. Still the authors show that in some of the scenarios this procedure performs better than FP-growth.

One way to potentially boost the algorithm of \cite{li2006searching} and \cite{yakout2007mining} by eliminating the potential repetitive mining overlap and fragmented information collection which are described above, is to replace the per itemset procedure by a per-level single call to the guided FP-growth procedure as follows. At each level, use the Apriori candidate-generation procedure and create a tree representing the candidates. Count the frequency of all the candidates by applying a single invocation of the guided FP-growth procedure with the candidate-representing {\it TIS-tree} as its guide. This procedure provides the advantage that no resources are wasted on counting the frequency of non-candidate itemsets. Note that in this case the compact representation of the dataset is created using an FP-tree instead of an itemset tree as used in \cite{yakout2007mining}. This adds two advantages, which are the performance and compactness due to building the tree from the frequent-items only from the start, and optimizing the order in which the items are used for building the tree. This comes at the expense of an additional pass through the database and less reusability of the tree for potential other required queries. 

Optimal classification-rule set were introduced in \cite{li2002mining} who showed that it has the same predictive power of the complete class-association rule set when used for building a classifier. The optimal rule-set includes only rules such that no rule built of the same consequence and a subset of the antecedent has an identical or better confidence. \cite{ghanem2014towards} presents a scheme for mining the optimal classification-rule set and then shows how to use this set for creating a robust classifier. \cite{ghanem2011edp} presents a distributed / parallel scheme for optimal-rule discovery, using a similar concepts.

Both \cite{ghanem2011edp} and \cite{ghanem2014towards} refer to \cite{kubat2003itemset}, \cite{yakout2007mining} and \cite{fournier2013meit} as the means for mining the next layer of rules. Therefore, the enhancement suggested above for \cite{yakout2007mining}, which involves replacing the repetitive and overlapping per rule invocation of a mining procedure by a per-layer invocation of the guided FP-growth procedure, is valid also for \cite{ghanem2011edp} and \cite{ghanem2014towards}, leading to potential improvement in their performance. Clearly, as demonstrated above, the GFP-growth algorithm will be of great use when dealing with imbalanced data. Specifically, we see a potential contribution by taking advantage of the ideas of \cite{li2002mining}, \cite{ghanem2011edp} and \cite{ghanem2014towards} and combine them with the contributions of the current paper, which may lead to better performance of the mining algorithms. 
\subsection{Incremental mining of frequent itemsets}
Many interesting use-cases for frequent-itemset mining involve datasets which are dynamically updated, with the most important case being addition of new data. Example domains in which such use-cases exist are e-commerce and data-streaming. Especially in such domains, fast response times are of importance as the insights need to reflect the current situation.

A data-set update may result in addition, deletion, or support-change of existing frequent itemsets. The straightforward way to get the updated information is to mine the entire amended dataset from scratch. However, significant time saving can be achieved when taking advantage of the previously gathered information.

Previous works present several different incremental algorithms aiming as much as possible to limit the processing to the newly added data in order to update the already-mined frequent itemsets. This allows minimization of the resources invested in this effort, such as time and memory.

A considerable amount of work was dedicated to incremental derivations of FP-growth. Examples for such incremental FP-growth derivation are \cite{ezeife2002mining}, \cite{gyorodi2003mining}, \cite{cheung2003incremental}, \cite{ma2004efficient}, \cite{hong2006fast}, \cite{lin2009pre}, \cite{pradeepini2010tree}, \cite{lin2010using}, \cite{totad2012batch}, and \cite{lin2013adminer}. The main challenge these papers attempt to address is the fact that an FP-tree does not contain information about non-frequent items appearing in the original-database transactions, which means that if an infrequent item becomes frequent due to the update, the original-database needs to be scanned again. \cite{totad2012batch} suggests creating a separate FP-tree (or a similar structure) for the new incremental dataset, and shows that this tree can be efficiently combined with the one representing the original dataset.

Most of these works assume that the resulting FP-tree needs to be mined again in order to find the new frequent item-sets created due to the update. \cite{ma2004efficient} suggests addressing this issue by maintaining information about each of the conditional FP-trees created during the FP-growth procedure in a hyper-tree in order to re-execute only those that are affected by the update.

We suggest using the guided FP-growth algorithm in order to take advantage of the frequent-itemset list already mined from the original dataset, and efficiently mine the updated frequent itemset list. The main idea is to perform guided mining of the (potentially huge) original FP-growth tree, focusing only on itemsets which may potentially become frequent. The itemsets potentially becoming frequent are those that are not frequent in the original dataset but are frequent on the incremental (new) dataset. This would enable updating the frequency of the itemsets already appearing in the original frequent-itemset tree.

\section{Summary}\label{s:con}
The main contribution of this work is the development of the GFP-growth (Guided FP-growth) algorithm, an FP-growth based algorithm for efficient multitude-targeted mining. The GFP-growth procedure serves for mining the support of multitude pre-specifed itemsets from an FP-tree and mines only the relevant parts of the FP-tree, in a single and partial FP-growth process. Our choice to base our algorithm on the original FP-growth was made due to the  popularity, availability, and many actual-use implementations of the  FP-growth algorithm. Indeed, our GFP-growth algorithm can take advantage of the many  FP-growth improvements that have been suggested in the literature. Such improvements will lead to additional time and memory costs reduction. More so, we have considered our contribution in the context of both targeted and constraint mining, which gives rise to further potential advantages. We demonstrated that the GFP-growth procedure is a very fast and generic tool, which can be applied to many different purposes, and provided theoretical results concerning its correctness.

An additional contribution of this work is the development of the Minority-Report Algorithm that uses the GFP-growth for boosting performance when generating the minority-class rules from imbalanced data. In that respect, we studied in detail the problem of mining minority-class rules from imbalanced data, a scenario that appears in many real-life domains such as medical applications, failure prediction, network and cyber security, and maintenance. We proved some theoretical properties of the Minority-Report Algorithm and demonstrated its performance gain using simulations and real data.

Finally, looking ahead for future research, we discussed how mining of optimal class association-rules
and how incremental itemset mining can potentially benefit from the use of the GFP-growth algorithm. Additional subjects for future work are exploring whether an efficient method for mining class association-rules from imbalanced data can be devised using the Minorty-Report algorithm, and whether the GFP-growth algorithm can be further optimized.

%\section{Section title}
%\label{sec:1}
%Text with citations \cite{RefB} and \cite{RefJ}.
%\subsection{Subsection title}
%\label{sec:2}
%as required. Don't forget to give each section
%and subsection a unique label (see Sect.~\ref{sec:1}).
%\paragraph{Paragraph headings} Use paragraph headings as needed.
%\begin{equation}
%\end{equation}

%\begin{figure}
%  \includegraphics{example.eps}
%\end{figure}
%\begin{figure*}
%  \includegraphics[width=0.75\textwidth]{example.eps}
%\end{figure*}
%\begin{table}
%\begin{tabular}{lll}
%\hline\noalign{\smallskip}
%\end{tabular}
%\end{table}

\paragraph{Acknowledgments}
%\begin{acknowledgements}
This work was supported by the Israeli Science Foundation grant number 387/15.
%\end{acknowledgements}

\bibliographystyle{unsrt}
\bibliography{bib}

%\begin{thebibliography}{}
% etc
%\end{thebibliography}

\end{document}